\documentclass[11pt]{article}
\usepackage{fullpage}
\usepackage[utf8]{inputenc}
\usepackage[final,pdfborder={0 0 0}]{hyperref}
\usepackage{upgreek}
\usepackage{cite}
\usepackage{amsmath,amsthm,amssymb,nicefrac,graphicx,lastpage,array}
\usepackage{multirow}
\usepackage{subfigure}

\usepackage[inline,nomargin]{fixme}
\fxsetface{inline}{\bfseries\color{red}}

\usepackage{tikz}
\usetikzlibrary{decorations.pathreplacing}
\usetikzlibrary{decorations.pathmorphing}
\usetikzlibrary{calc}
\usetikzlibrary{shapes.geometric}
\usetikzlibrary{shapes.misc}
\usetikzlibrary{positioning}
\usetikzlibrary{patterns}
\usetikzlibrary{fit}

\usetikzlibrary{intersections}

\pgfdeclarelayer{background}
\pgfdeclarelayer{foreground}
\pgfsetlayers{background,main,foreground}

\tikzstyle{every picture}+=[remember picture]
\tikzstyle{defnode}=[solid,thin,draw,circle,minimum size=4pt,inner sep=0pt]
\tikzstyle{edgelab}=[draw=none,fill=white,font=\scriptsize,inner sep=0.5pt, text height=4pt, text depth=0pt,minimum size=5pt]
\tikzstyle{lightedge}=[draw,thin,dotted,semithick]
\tikzstyle{heavyedge}=[draw,solid,very thick]

\usepackage[noend]{algorithmic}

\usepackage{algorithm}

\newtheorem{theorem}{Theorem}
\newtheorem{lemma}{Lemma}
\newtheorem{corollary}{Corollary}
\newtheorem{example}{Example}

\renewcommand{\epsilon}{\upvarepsilon}

\newcommand{\pref}{\operatorname{pref}}
\newcommand{\suff}{\operatorname{suff}}

\newcommand{\weight}{\operatorname{weight}}
\newcommand{\height}{\operatorname{height}}

\newcommand{\lightdepth}{\operatorname{lightdepth}}
\newcommand{\lightheight}{\operatorname{lightheight}}

\newcommand{\lightstrings}{\operatorname{lightstrings}}

\newcommand{\ind}{\operatorname{pos}}

\newcommand{\lcp}{\operatorname{\textsc{lcp}}}

\newcommand{\wildcard}{\ast}
\newcommand{\gap}[2]{\mbox{$\wildcard \{ #1, #2 \}$}}

\title{String Indexing for Patterns with Wildcards\thanks{Preliminary version appeared in \emph{Proceedings of the 13th Scandinavian Symposium and Workshops on Algorithm Theory}. Lecture Notes in Computer Science, vol. 7357, pp. 283--294, Springer 2012.}}
\author{Philip Bille ~~~ Inge Li Gørtz\thanks{Supported by a grant from the Danish Council for Independent Research $\vert$ Natural Sciences} ~~~ Hjalte Wedel Vildhøj ~~~ Søren Vind\\[0.2em]
Technical University of Denmark, DTU Informatics\\
\texttt{\small{\{phbi,ilg,hwvi,sjvi\}@imm.dtu.dk}}}
\date{}
\begin{document}
	
\maketitle

\begin{abstract}
\noindent We consider the problem of indexing a string $t$ of length $n$ to report the occurrences of a query pattern $p$ containing $m$ characters and $j$ wildcards. Let $occ$ be the number of occurrences of $p$ in $t$, and $\sigma$ the size of the alphabet. We obtain the following results.
\begin{itemize}
	\item A linear space index with query time $O(m+\sigma^j \log \log n + occ)$. 
	This significantly improves the previously best known linear space index by Lam~et~al.~[ISAAC~2007], which requires query time $\Theta(jn)$ in the worst case.
	\item An index with query time $O(m+j+occ)$ using space $O(\sigma^{k^2} n \log^k \log n)$, where $k$ is the maximum number of wildcards allowed in the pattern. This is the first non-trivial bound with this query time.
 	\item A time-space trade-off, generalizing the index by Cole~et~al.~[STOC~2004].
\end{itemize}
We also show that these indexes can be generalized to allow variable length gaps in the pattern. Our results are obtained using a novel combination of well-known and new techniques, which could be of independent interest.
\end{abstract}

\section{Introduction}
The \emph{string indexing problem} is to build an index for a string $t$ such that the occurrences of a query pattern $p$ can be reported. The classic suffix tree data structure~\cite{Weiner1973} combined with perfect hashing~\cite{fredman1984storing} gives a linear space solution for string indexing with optimal query time, i.e., an $O(n)$ space data structure that supports queries in $O(m+occ)$ time, where $occ$ is the number of occurrences of $p$ in $t$.

Recently, various extensions of the classic string indexing problem that allow errors or wildcards (also known as gaps or don't cares) have been studied \cite{Cole2004,Lam2007,Tsur2010,tam2009succinct,Chan2011,Maas2007,Baeza-yates}. In this paper, we focus on one of the most basic of these extensions, namely, \emph{string indexing for patterns with wildcards}. In this problem, only the pattern contains wildcards, and the goal is to report all occurrences of $p$ in $t$, where a wildcard is allowed to match any character in $t$.

String indexing for patterns with wildcards finds several natural applications in large-scale data processing areas such as information retrieval, bioinformatics, data mining, and internet traffic analysis.  For instance in bioinformatics, the PROSITE data base~\cite{Hofmann1999,BB1994} supports searching for protein patterns containing wildcards. 

Despite significant interest in the problem and its many variations, most of the basic questions remain unsolved. We introduce three new indexes and obtain several new bounds for string indexing with wildcards in the pattern. If the index can handle patterns containing an unbounded number of wildcards, we call it an \emph{unbounded wildcard index}, otherwise we refer to the index as a \emph{$k$-bounded wildcard index}, where $k$ is the maximum number of wildcards allowed in $p$. Let $n$ be the length of the indexed string $t$, and $\sigma$ be the size of the alphabet. We define $m$ and $j$ to be the number of characters and wildcards in $p$, respectively. Consequently, the length of $p$ is $m+j$. We show that,
\begin{itemize}
	\item There is an unbounded wildcard index with query time $O(m+\sigma^j \log \log n + occ)$ using linear space. 
	This significantly improves the previously best known linear space index by Lam~et~al.~\cite{Lam2007}, which requires query time $\Theta(jn)$ in the worst case. Compared to the index by Cole~et~al.~\cite{Cole2004} having the same query time, we improve the space usage by a factor $\log n$.
	\item There is a $k$-bounded wildcard index with query time $O(m+j+occ)$ using space $O(\sigma^{k^2} n \log^k \log n)$. This is the first non-trivial space bound with this query time.
 	\item There is a time-space trade-off for $k$-bounded wildcard indexes. This trade-off generalizes the index described by Cole~et~al.~\cite{Cole2004}.
\end{itemize}

\noindent Furthermore, we generalize these indexes to support \emph{variable length gaps} in the pattern.

\subsection{Previous Work}\label{sec:previouswork}
Exact string matching has been generalized with error bounds in many different ways. In particular, allowing matches within a bounded hamming or edit distance, known as approximate string matching, has been subject to much research~\cite{landau1986efficient,landau1989fast,sahinalp1996efficient,cole1998approximate,Oliveira2006,Tsur2010,Chan2011,Maas2007,Cole2004,Baeza-yates,galil1986improved,amir2000faster}. Another generalization was suggested by Fischer~and~Paterson~\cite{Fischer1974}, allowing wildcards in the text or pattern.

Work on the wildcard problem has mostly focused on the non-indexing variant, where the string $t$ is not preprocessed in advance~\cite{Fischer1974,cole2002verifying,clifford2007simple,kalai2002efficient,chen2006efficient,bille2010string}. Some solutions to the indexing problem consider the case where wildcards appear only in the indexed string~\cite{tam2009succinct} or in both the string and the pattern~\cite{Cole2004,Lam2007}. 

In the following, we summarize the known indexes that support wildcards in the pattern only.
We focus on the case where $k > 1$, since for $k = 0$ the problem is classic string indexing. For $k=1$, Cole~et~al.~\cite{Cole2004} describe a selection of specialized solutions. However, these solutions do not generalize to larger $k$.

Several simple solutions to the problem exist for $k > 1$. Using a suffix tree $T$ for $t$~\cite{Weiner1973}, we can find all occurrences of $p$ in a top-down traversal starting from the root. When we reach a wildcard character in $p$ in location $\ell \in T$, the search branches out, consuming the first character on all outgoing edges from $\ell$. This gives an unbounded wildcard index using $O(n)$ space with query time $O(\sigma^j m + occ)$, where $occ$ is the total number of occurrences of $p$ in $t$. Alternatively, we can build a compressed trie storing all possible modifications of all suffixes of $t$ containing at most $k$ wildcards. This gives a $k$-bounded wildcard index using $O(n^{k+1})$ space with query time $O(m+j+occ)$.

In 2004, Cole~et~al.~\cite{Cole2004} gave an elegant $k$-bounded wildcard index using $O(n \log^k n)$ space and with $O(m+2^j \log\log n + occ)$ query time. For sufficiently small values of $j$ this significantly improves the previous bounds. The key components in this solution are a new data structure for \emph{longest common prefix (LCP) queries} and a \emph{heavy path decomposition}~\cite{harel1984fast} of the suffix tree for the text $t$. Given a pattern $p$, the LCP data structure supports efficient insertion of all suffixes of $p$ into the suffix tree for $t$, such that subsequent longest common prefix queries between any pair of suffixes from $t$ and $p$ can be answered in $O(\log \log n)$ time. This is where the $\log \log n$ term in the query time comes from. The heavy path decomposition partitions the suffix tree into disjoint \emph{heavy paths} such that any root-to-leaf path contains at most a logarithmic number of heavy paths. Cole~et~al.~\cite{Cole2004} show how to reduce the size of the index by only creating additional wildcard tries for the off-path subtries. This leads to the $O(n \log^k n)$ space bound. Secondly, using the new tries, the top-down search branches at most twice for each wildcard, leading to the $2^j$ term in the query time. Though Cole~et~al.~\cite{Cole2004} did not consider unbounded wildcard indexes, the technique can be extended to this case by using only the LCP data structure and omitting the additional wildcard tries. This leads to an unbounded wildcard index with query time $O(m + \sigma^j \log\log n + occ)$ using space $O(n\log n)$.

The solutions described by Cole~et~al.~\cite{Cole2004} all have bounds which are exponential in the number of wildcards in the pattern. Very recently, Lewenstein~\cite{lewenstein2011indexing} used similar techniques to improve the bounds to be exponential in the number of \emph{gaps} in the pattern (a gap is a maximal substring of consecutive wildcards). Assuming that the pattern contains at most $g$ gaps each of size at most $G$, Lewenstein obtains a bounded index with query time $O(m+2^\gamma \log\log n + occ)$ using space $O(n (G^2 \log n)^g)$, where $\gamma \leq g$ is the number of gaps in the pattern.

A different approach was taken by Iliopoulos~and~Rahman~\cite{Iliopoulos2007}, who describe an unbounded wildcard index using linear space. For a pattern $p$ consisting of strings $p_0,p_1,\ldots,p_j$ (subpatterns) interleaved by $j$ wildcards, the query time of the index is $O(m + \sum_{i=0}^j occ(p_i, t))$, where $occ(p_i, t)$ denotes the number of matches of $p_i$ in $t$. This was later improved by Lam~et~al.~\cite{Lam2007} with an index that determines complete matches by first identifying potential matches of the subpatterns in $t$ and subsequently verifying each possible match for validity using interval stabbing on the subpatterns. Their solution is an unbounded wildcard index with query time $O \left ( m + j \min_{0 \leq i \leq j} occ(p_i, t) \right )$ using linear space. However, both of these solutions have a worst case query time of $\Theta(jn)$, since there may be $\Theta(n)$ matches for a subpattern, but no matches of $p$.
\autoref{resultsummary} summarizes the existing solutions for the problem in relation to our results.

\begin{table}[t]\centering	\scriptsize
\begin{tabular}{| l | l | l | l r |}
	\hline
	\textbf{Type} & \textbf{Query Time} & \textbf{Space} & \textbf{Solution} & \\
	\hline
	\multirow{5}{*}{Unbounded} & $O(m + \sum_{i=0}^j occ(p_i, t))$ & $O(n)$ & Iliopoulos~and~Rahman & \cite{Iliopoulos2007} \\
	& $O(m+j \min_{0 \leq i \leq j} occ(p_i, t))$ & $O(n)$ & Lam~et~al. & \cite{Lam2007} \\
	& \boldmath{$O(\sigma^j m + occ)$} & \boldmath{$O(n)$} & \textbf{Simple suffix tree index} & \boldmath{$\dagger$} \\
	& \boldmath{$O(m+\sigma^j \log\log n + occ)$} & \boldmath{$O(n)$} & \textbf{ART decomposition} & \boldmath{$\dagger$} \\
	& $O(m+\sigma^j \log\log n + occ)$ & $O(n \log n)$ & Cole~et~al. & \cite{Cole2004} \\
	\hline	
	\multirow{4}{*}{$k$-bounded} & \boldmath{$O(m+\beta^j \log\log n + occ)$} & \boldmath{$O(n \log n \log^{k-1}_\beta n)$} & \textbf{Heavy \boldmath{$\alpha$}-tree decomposition} & \boldmath{$\dagger$} \\
	& $O(m+2^j \log\log n + occ)$ & $O(n \log^k n)$ & Cole~et~al. & \cite{Cole2004} \\
	& \boldmath{$O(m + j + occ)$} & \boldmath{$O(n \sigma^{k^2} \log^k \log n)$} & \textbf{Special index for \boldmath{$m < \sigma^k \log \log n$}} & \boldmath{$\dagger$} \\
	& \boldmath{$O(m + j + occ)$} & \boldmath{$O(n^{k+1})$} & \textbf{Simple linear time index} & \boldmath{$\dagger$} \\
	\hline
\end{tabular}
\caption{$\dagger$ = presented in this paper. The term $occ(p_i, t)$ denotes the number of matches of $p_i$ in $t$ and is $\Theta (n)$ in the worst case.} \label{resultsummary}
\end{table}

The unbounded wildcard index by Iliopoulos and Rahman~\cite{Iliopoulos2007} was the first index to achieve query time linear in $m$ while using $O(n)$ space. Recently, Chan~et~al.~\cite{Chan2011} considered the related problem of obtaining a $k$-mismatch index supporting queries in time linear in $m$ and using $O(n)$ space. They describe an index with a query time of $O(m + (\log n)^{k(k+1)} \log\log n + occ)$. However, this bound assumes a constant-size alphabet and a constant number of errors. In this paper we make no assumptions on the size of these parameters.

\subsection{Our Results}\label{sec:ourresults}
Our main contribution is three new wildcard indexes.

\begin{theorem}\label{optimalspaceindex}
Let $t$ be a string of length $n$ from an alphabet of size $\sigma$. There is an unbounded wildcard index for $t$ using $O(n)$ space. The index can report the occurrences of a pattern with $m$ characters and $j$ wildcards in time $O(m+\sigma^j \log \log n + occ)$.
\end{theorem}
Compared to the solution by Cole~et~al.~\cite{Cole2004}, we obtain the same query time while reducing the space usage by a factor $\log n$. We also significantly improve upon the previously best known linear space index by Lam~et~al.~\cite{Lam2007}, as we match the linear space usage while improving the worst-case query time from $\Theta(jn)$ to $O(m+\sigma^j \log \log n + occ)$ provided $j \leq \log_\sigma n$. Our solution is faster than the simple suffix tree index for $m = \Omega(\log \log n)$. Thus, for sufficiently small $j$ we improve upon the previously known unbounded wildcard indexes.

The main idea of the solution is to combine an ART decomposition~\cite{Alstrup1998} of the suffix tree for $t$ with the LCP data structure. The suffix tree is decomposed into a number of logarithmic-sized bottom trees and a single top tree. We introduce a new variant of the LCP data structure for use on the bottom trees, which supports queries in logarithmic time and linear space. The logarithmic size of the bottom trees leads to LCP queries in time $O(\log\log n)$. On the top tree we use the LCP data structure by Cole~et~al.~\cite{Cole2004} to answer queries in time $O(\log\log n)$. The number of LCP queries performed during a search for $p$ is $O(\sigma^j)$, yielding the $\sigma^j\log\log n$ term in the query time. The reduced size of the top tree causes the index to be linear in size.

\begin{theorem}\label{wst-index}
Let $t$ be a string of length $n$ from an alphabet of size $\sigma$. For $2 \leq \beta < \sigma$, there is a $k$-bounded wildcard index using $O ( n\log(n) \log_\beta^{k-1} n )$ space. The index can report the occurrences in $t$ of a pattern with $m$ characters and $j \leq k$ wildcards in time $O (m + \beta^j \log\log n + occ )$. 
\end{theorem}
The theorem provides a time-space trade-off for $k$-bounded wildcard indexes. Compared to the index by Cole~et~al.~\cite{Cole2004}, we reduce the space usage by a factor $\log^{k-1} \beta$ by increasing the branching factor from $2$ to $\beta$. For $\beta=2$ the index is identical to the index by Cole~et~al. The result is obtained by generalizing the wildcard index described by Cole~et~al. We use a \emph{heavy $\alpha$-tree decomposition}, which is a new technique generalizing the classic heavy path decomposition by Harel~and~Tarjan~\cite{harel1984fast}. This decomposition could be of independent interest. We also show that for $\beta=1$ the same technique yields an index with query time $O(m+j+occ)$ using space $O(nh^k)$, where $h$ is the height of the suffix tree for $t$.

\begin{theorem}\label{optimaltimeindex}
Let $t$ be a string of length $n$ from an alphabet of size $\sigma$. There is a $k$-bounded wildcard index for $t$ using $O(\sigma^{k^2} n \log^k \log n)$ space. The index can report the occurrences of a pattern with $m$ characters and $j \leq k$ wildcards in time $O(m+j+occ)$.
\end{theorem}
To our knowledge this is the first linear time index with a non-trivial space bound. The result improves upon the space usage of the simple linear time index when $\sigma^k < n/\log \log n$. To achieve this result, we use the $O(nh^k)$ space index to obtain a black-box reduction that can produce a linear time index from an existing index. The idea is to build the $O(nh^k)$ space index with support for short patterns, and query another index if the pattern is long. This technique is closely related to the concept of \emph{filtering search} introduced by Chazelle~\cite{chazelle1986} and has previously been applied for indexing problems~\cite{Bille2011,Chan2011}. The theorem follows from applying the black-box reduction to the index of \autoref{optimalspaceindex}.\\

\subsubsection{Variable Length Gaps}
\noindent We also show that the three indexes support searching for query patterns with \emph{variable length gaps}, i.e., patterns of the form $p=p_0 \, \wildcard \! \{a_1,b_1\} \, p_1  \wildcard \! \{a_2,b_2\} \ldots \wildcard \! \{a_j,b_j\} \, p_j$, where $ \wildcard \{a_i,b_i\}$ denotes a variable length gap that matches an arbitrary substring of length between $a_i$ and $b_i$, both inclusive.

String indexing for patterns with variable length gaps has applications in information retrieval, data mining and computational biology~\cite{fredriksson2008efficient,fredriksson2009nested,myers1996approximate,mehldau1993system,navarro2003fast}. In particular, the PROSITE data base~\cite{Hofmann1999,BB1994} uses patterns with variable length gaps to identify and classify protein sequences. The problem is a generalization of string indexing for patterns with wildcards, since a wildcard $\wildcard$ is equivalent to the variable length gap $\gap{1}{1}$. Variable length gaps are also known as \emph{bounded wildcards}, as a variable length gap $\gap{a_i}{b_i}$ can be regarded as a bounded sequence of wildcards. 

String indexing for patterns with variable length gaps is equivalent to string indexing for patterns with wildcards, with the addition of allowing \emph{optional wildcards} in the pattern. An optional wildcard matches any character from $\Sigma$ or the empty string, i.e., an optional wildcard is equivalent to the variable length gap $\gap{0}{1}$. Conversely, we may also consider a variable length gap $\gap{a_i}{b_i}$ as $a_i$ consecutive wildcards followed by $b_i - a_i$ consecutive optional wildcards.

Lam~et~al.~\cite{Lam2007} introduced optional wildcards in the pattern and presented a variant of their solution for the string indexing for patterns with wildcards problem. The idea is to determine potential matches and verify complete matches using interval stabbing on the possible positions for the subpatterns. This leads to an unbounded optional wildcard index with query time $O(m + B j \min_{0 \leq i \leq j} occ(p_i, t))$ and space usage $O(n)$. Here $B = \sum_{i=1}^j b_i$ and $occ(p_i, t)$ denotes the number of matches of $p_i$ in $t$, and since $occ(p_i, t) = \Theta(n)$ in the worst case, the worst case query time is $\Theta(Bjn)$.
Recently, Lewenstein~\cite{lewenstein2011indexing} considered the special case where the pattern contains at most $g$ gaps and $a_i=b_i \leq G$ for all $i$, i.e., the gaps are non-variable and of length at most $G$. Using techniques similar to those by Cole~et~al.~\cite{Cole2004}, he gave a bounded index with query time $O(m+2^\gamma \log\log n + occ)$ using space $O(n (G^2 \log n)^g)$, where $\gamma \leq g$ is the number of gaps in the pattern.

The related \emph{string matching with variable length gaps problem}, where the text may not be preprocessed in advance, has recieved some research attention recently~\cite{bille2010string, RILMS2006, MPVZ2005, navarro2003fast, FG2008}. However, none of the results and techniques developed for this problem appear to lead to non-trivial bounds for the indexing problem.

\paragraph{Our Results for Variable Length Gaps}\label{vlg-ourres}
To introduce our results we let $A = \sum_{i=1}^j a_i$ and $B = \sum_{i=1}^j b_i$ denote the sum of the lower and upper bounds on the variable length gaps in $p$, respectively. Hence $A$ and $B-A$ denote the number of normal and optional wildcards in $p$, respectively. A wildcard index with support for optional wildcards is called an \emph{optional wildcard index}. As for wildcard indexes, we distinguish between bounded and unbounded optional wildcard indexes. A $(k,o)$-bounded optional wildcard index supports patterns containing $A \leq k$ normal wildcards and $B-A \leq o$ optional wildcards. An unbounded optional wildcard index supports patterns with no restriction on the number of normal and optional wildcards.

To accommodate for variable length gaps in the pattern, we only need to modify the way in which the wildcard indexes are searched, leading to the following new theorems. The proofs are given in \autoref{sec:vlg}.

\begin{theorem}\label{vlg-optimalspaceindex}
Let $t$ be a string of length $n$ from an alphabet of size $\sigma$. There is an unbounded optional wildcard index for $t$ using $O(n)$ space. The index can report the occurrences of a pattern with $m$ characters, $A$ wildcards and $B-A$ optional wildcards in time $O(m+ 2^{B-A} \sigma^B \log \log n + occ)$, where $A = \sum_{i=1}^j a_i$ and $B = \sum_{i=1}^j b_i$.
\end{theorem} 

\begin{theorem}\label{vlg-wst-index}
Let $t$ be a string of length $n$ from an alphabet of size $\sigma$. For $2 \leq \beta < \sigma$, there is a $(k,o)$-bounded optional wildcard index for $t$ using $O \bigl( n\log(n) \log_\beta^{k+o-1} n \bigr )$ space. The index can report the occurrences of a pattern with $m$ characters, $A \leq k$ wildcards and $B-A \leq o$ optional wildcards in time $O \left (m + 2^{B-A} \beta^B \log\log n + occ \right )$, where $A = \sum_{i=1}^j a_i$ and $B = \sum_{i=1}^j b_i$.
\end{theorem}

\begin{theorem}\label{vlg-optimaltimeindex}
Let $t$ be a string of length $n$ from an alphabet of size $\sigma$. There is a $(k,o)$-bounded optional wildcard index for $t$ using $O(\sigma^{(k+o)^2} n \log^{k+o} \log n)$ space. The index can report the occurrences of a pattern with $m$ characters, $A \leq k$ wildcards and $B-A \leq o$ optional wildcards in time $O(2^{B-A}(m+B)+occ)$, where $A = \sum_{i=1}^j a_i$ and $B = \sum_{i=1}^j b_i$.
\end{theorem}

\noindent These results completely generalize our previous solutions, since if the query pattern only contains variable length gaps of the form $\gap{1}{1}$, the problem reduces to string indexing for patterns with wildcards. In that case $A=B=j$ and we obtain exactly \autoref{optimalspaceindex}, \autoref{wst-index} and \autoref{optimaltimeindex}.

Compared to the only known index for the problem by Lam~et~al.~\cite{Lam2007}, \autoref{vlg-optimalspaceindex} gives an unbounded optional wildcard index that matches the $O(n)$ space usage, but improves the worst-case query time from $\Theta(Bjn)$ to $O \left (m + 2^{B-A}\sigma^B \log\log n + occ \right )$, provided that $B \leq \log_\sigma \sqrt{nj}$.

\section{Preliminaries}\label{sec:preliminaries}
We introduce the following notation. Let $p=p_0 \wildcard p_1  \wildcard \ldots \wildcard p_j$ be a pattern consisting of $j+1$ strings $p_0,p_1,\ldots,p_{j} \in \Sigma^*$ (subpatterns) interleaved by $j \leq k$ wildcards. The substring starting at position $l \in \{1,\ldots,n\}$ in $t$ is an occurrence of $p$ if and only if each subpattern $p_i$ matches the corresponding substring in $t$.  That is,
\[
p_i = t\left [l+i+\sum_{r=0}^{i-1} |p_r|, l+i-1+\sum_{r=0}^{i} |p_r| \right ] \quad \text{for} \quad  i=0,1,\ldots,j \; ,
\]
where $t[i,j]$ denotes the substring of $t$ between indices $i$ and $j$, both inclusive.
We define $t[i,j]=\epsilon$ for $i>j$, $t[i,j]=t[1,j]$ for $i<1$ and $t[i,j]=t[i,|t|]$ for $j>|t|$. Furthermore $m=\sum_{r=0}^j |p_r|$ is the number of characters in $p$, and we assume without loss of generality that $m > 0$ and $k>0$. 

Let $\pref_i(t) = t[1,i]$ and $\suff_i(t) = t[i,n]$ denote the prefix and suffix of $t$ of length $i$ and $n-i+1$, respectively. Omitting the subscripts, we let $\pref(t)$ and $\suff(t)$ denote the set of all non-empty prefixes and suffixes of $t$, respectively. We extend the definitions of prefix and suffix to sets of strings $S \subseteq \Sigma^*$ as follows.
\begin{align*}
\pref_i(S) &= \{ \pref_i(x) \mid x \in S\} \qquad &\suff_i(S) &= \{ \suff_i(x) \mid x \in S\} \\
\pref(S) &= \bigcup_{x \in S} \pref(x) \qquad &\suff(S) &= \bigcup_{x \in S} \suff(x)
\end{align*}
A set of strings $S$ is \emph{prefix-free} if no string in $S$ is a prefix of another string in $S$. Any string set $S$ can be made prefix-free by appending the same unique character $\$ \notin \Sigma$ to each string in $S$.

\subsection{Trees and Tries}
For a tree $T$, the root is denoted $\operatorname{root}(T)$, while $\height(T)$ is the number of edges on a longest path from $\operatorname{root}(T)$ to a leaf of $T$. A compressed trie $T(S)$ is a tree storing a prefix-free set of strings $S \subset \Sigma^*$. The edges are labeled with substrings of the strings in $S$, such that a path from the root to a leaf corresponds to a unique string in $S$. All internal vertices (except the root) have at least two children, and all labels on the outgoing edges of a vertex have different initial characters.

 A \emph{location} $\ell \in T(S)$ may refer to either a vertex or a position on an edge in $T(S)$. Formally, $\ell = (v,s)$ where $v$ is a vertex in $T(S)$ and $s \in \Sigma^*$ is a prefix of the label on an outgoing edge of $v$. If $s=\epsilon$, we also refer to $\ell$ as an \emph{explicit vertex}, otherwise $\ell$ is called an \emph{implicit vertex}. There is a one-to-one mapping between locations in $T(S)$ and unique prefixes in $\pref(S)$. The prefix $x \in \pref(S)$ corresponding to a location $\ell \in T(S)$ is obtained by concatenating the edge labels on the path from $\operatorname{root}(T(S))$ to $\ell$. Consequently, we use $x$ and $\ell$ interchangeably, and we let $|\ell|=|x|$ denote the length of $x$. Since $S$ is assumed prefix-free, each leaf of $T(S)$ is a string in $S$, and conversely. The \emph{suffix tree} for $t$ denotes the compressed trie over all suffixes of $t$, i.e., $T(\suff(t))$. We define $T_\ell(S)$ as the subtrie of $T(S)$ rooted at $\ell$. That is, $T_\ell(S)$ contains the suffixes of strings in $T(S)$ starting from $\ell$. Formally, $T_\ell(S) = T(S_\ell)$, where
\[
S_\ell = \left \{ \suff_{|\ell|}(x) \mid x \in S \wedge \pref_{|\ell|}(x) = \ell \right \} \; .
\]
\subsection{Heavy Path Decomposition}
For a vertex $v$ in a rooted tree $T$, we define $\weight(v)$ to be the number of leaves in $T_v$, where $T_v$ denotes the subtree rooted at $v$. We define $\weight(T) = \weight(\operatorname{root}(T))$. The \emph{heavy path decomposition} of $T$, introduced by Harel~and~Tarjan~\cite{harel1984fast}, classifies each edge as either \emph{light} or \emph{heavy}. For each vertex $v \in T$, we classify the edge going from $v$ to its child of maximum weight (breaking ties arbitrarily) as heavy. The remaining edges are light. This construction has the property that on a path from the root to any vertex, $O(\log(\weight(T)))$ heavy paths are traversed.
For a heavy path decomposition of a compressed trie $T(S)$, we assume that the heavy paths are extended such that the label on each light edge contains exactly one character.

\section{The LCP Data Structure}\label{sec:lcp}

Cole~et~al.~\cite{Cole2004} introduced the the \emph{Longest Common Prefix (LCP) data structure}, which provides a way to traverse a compressed trie without tracing the query string one character at a time. In this section we give a brief, self-contained description of the data structure and show a new property that is essential for obtaining \autoref{optimalspaceindex}.

The LCP data structure stores a collection of compressed tries $T(C_1),
 T(C_2),$ $\ldots, T(C_q)$ over the string sets $C_1,C_2,\ldots,C_q \subset \Sigma^*$. Each $C_i$ is a set of substrings of the indexed string $t$. The purpose of the LCP data structure is to support LCP queries
\begin{description}
\item $\lcp(x,i,\ell)$: Returns the location in $T(C_i)$ where the search for the string $x \in \Sigma^*$ stops when starting in location $\ell \in T(C_i)$.
\end{description}
If $\ell$ is the root of $T(C_i)$, we refer to the above LCP query as a \emph{rooted LCP query}. Otherwise the query is called an \emph{unrooted LCP query}. In addition to the compressed tries $T(C_1),\ldots,T(C_q)$, the LCP data structure also stores the suffix tree for $t$, denoted $T(C)$ where $C = \suff(t)$. The following lemma is implicit in the paper by Cole~et~al.~\cite{Cole2004}.

\begin{lemma}[Cole~et~al.~\cite{Cole2004}]\label{lcp-ds}
Provided $x$ has been preprocessed in time $O(|x|)$, the LCP data structure can answer rooted LCP queries on $T(C_i)$ for any suffix of $x$ in time $O(\log\log |C|)$ using space $O(|C| + \sum_{i=1}^q |C_i|)$. Unrooted LCP queries on $T(C_i)$ can be performed in time $O(\log\log |C|)$ using $O(|C_i|\log |C_i|)$ additional space.
\end{lemma}
We extend the LCP data structure by showing that support for slower unrooted LCP queries on a compressed trie $T(C_i)$ can be added using linear additional space.
\begin{lemma}\label{lcp-ext}
Unrooted LCP queries on $T(C_i)$ can be performed in time $O(\log|C_i| + \log\log |C|)$ using $O(|C_i|)$ additional space.
\end{lemma}

\begin{proof}
We initially create a heavy path decomposition for all compressed tries $T(C_1), \ldots, T(C_q)$. The search path for $x$ starting in $\ell$ traverses a number of heavy paths in $T(C_i)$. Intuitively, an unrooted LCP query can be answered by following the $O(\log |C_i|)$ heavy paths that the search path passes through. For each heavy path, the next heavy path can be identified in constant time. On the final heavy path, a predecessor query is needed to determine the exact location where the search path stops.

For a heavy path $H$, we let $h$ denote the distance that the search path for $x$ follows $H$. Cole~et~al.~\cite{Cole2004} showed that $h$ can be determined in constant time by performing nearest common ancestor queries on $T(C)$. To answer $\lcp(x, i, \ell)$ we identify the heavy path $H$ of $T(C_i)$ that $\ell$ is part of and compute the distance $h$ as described by Cole~et~al. If $x$ leaves $H$ on a light edge, indexing distance $h$ into $H$ from $\ell$ yields an explicit vertex $v$. At $v$, a constant time lookup for $x[h+1]$ determines the light edge on which $x$ leaves $H$. Since the light edge has a label of length one, the next location $\ell^\prime$ on that edge is the root of the next heavy path. We continue the search for the remaining suffix of $x$ from $\ell^\prime$ recursively by a new unrooted LCP query $\lcp(\suff_{h+2}(x),i,\ell^\prime)$. If $H$ is the heavy path on which the search for $x$ stops, the location at distance $h$ (i.e., the answer to the original LCP query) is not necessarily an explicit vertex, and may not be found by indexing into $H$. In that case a predecessor query for $h$ is performed on $H$ to determine the preceding explicit vertex and thereby the location $\lcp(x,i,\ell)$. Answering an unrooted LCP query entails at most $\log |C_i|$ recursive steps, each taking constant time. The final recursive step may require a predecessor query taking time $O(\log\log |C|)$. Consequently, an unrooted LCP query can be answered in time $O(\log|C_i| + \log\log |C|)$ using $O(|C_i|)$ additional space to store the predecessor data structures for each heavy path.
\end{proof}

\section{An Unbounded Wildcard Index Using Linear Space}\label{sec:optimalspaceindex}
In this section we show how to obtain \autoref{optimalspaceindex} by applying an ART decomposition on the suffix tree for $t$ and storing the top and bottom trees in the LCP data structure. 

\subsection{ART Decomposition}\label{artdecomposition}
The ART decomposition introduced by Alstrup~et~al.~\cite{Alstrup1998} decomposes a tree into a single \emph{top tree} and a number of \emph{bottom trees}. The construction is defined by two rules:
\begin{enumerate}
	\item A bottom tree is a subtree rooted in a vertex of minimal depth such that the subtree contains no more than $\chi$ leaves. 
	\item Vertices that are not in any bottom tree make up the top tree.
\end{enumerate}

The decomposition has the following key property.

\begin{lemma}[Alstrup et al.~\cite{Alstrup1998}]\label{lem:art-top}
	The ART decomposition with parameter $\chi$ for a rooted tree $T$ with $n$ leaves produces a top tree with at most $\frac{n}{\chi + 1}$ leaves.
\end{lemma}

\subsection{Obtaining the Index}\label{optimalspace-structure}
Applying an ART decomposition on $T(\suff(t))$ with $\chi = \log n$, we obtain a top tree $T^\prime$ and a number of bottom trees $B_1, B_2, \ldots, B_q$ each of size at most $\log n$. From \autoref{lem:art-top}, $T^\prime$ has at most $\frac{n}{\log n}$ leaves and hence $O(\frac{n}{\log n})$ vertices since $T^\prime$ is a compressed trie.

To facilitate the search, the top and bottom trees are stored in an LCP data structure, noting that these compressed tries only contain substrings of $t$. Using \autoref{lcp-ext}, we add support for unrooted $O(\log \chi + \log\log n)  = O(\log \log n)$ time LCP queries on the bottom trees using $O(n)$ additional space in total. For the top tree we apply \autoref{lcp-ds} to add support for unrooted LCP queries in time $O(\log\log n)$ using $O(\frac{n}{\log n} \log \frac{n}{\log n}) = O(n)$ additional space. Since the branching factor is not reduced, $O(\sigma^i)$ LCP queries, each taking time $O(\log \log n)$, are performed for the subpattern $p_i$. This concludes the proof of \autoref{optimalspaceindex}.

\section{A Time-Space Trade-Off for $k$-Bounded Wildcard Indexes}\label{sec:tradeoff}
In this section we will show \autoref{wst-index}. We first introduce the necessary constructions.

\subsection{Heavy $\alpha$-Tree Decomposition}\label{alphaheavysec}
The \emph{heavy $\alpha$-tree decomposition} is a generalization of the well-known heavy path decomposition introduced by Harel~and~Tarjan~\cite{harel1984fast}. The purpose is to decompose a rooted tree $T$ into a number of \emph{heavy trees} joined by light edges, such that a path to the root of $T$ traverses at most a logarithmic number of heavy trees. For use in the construction, we define a \emph{proper weight function} on the vertices of $T$, to be a function satisfying
\(
\weight(v) \geq \sum_{w \text{ child of }v} \weight(w) \; .
\)
Observe that using the number of vertices or the number of leaves in the subtree rooted at $v$ as the weight of $v$ satisfies this property. The decomposition is then constructed by classifying edges in $T$ as being heavy or light according to the following rule.
For every vertex $v \in T$, the edges to the $\alpha$ heaviest children of $v$ (breaking ties arbitrarily) are heavy, and the remaining edges are light.
For $\alpha=1$ this results in a heavy path decomposition. Given a heavy $\alpha$-tree decomposition of $T$, we define $\lightdepth_\alpha(v)$ to be the number of light edges on a path from the vertex $v \in T$ to the root of $T$. The key property of this construction is captured by the following lemma.
\begin{lemma}\label{lightdepth}
For any vertex $v$ in a rooted tree $T$ and $\alpha > 0$
\[
\lightdepth_\alpha(v) \leq \log_{\alpha+1} \weight(\operatorname{root}(T))
\]
\end{lemma}
\begin{proof}
Consider a light edge from a vertex $v$ to its child $w$. We prove that $\weight(w) \leq \frac{1}{\alpha+1} \weight(v)$, implying that $\lightdepth_\alpha(v) \leq \log_{\alpha+1} \weight(\operatorname{root}(T))$. To obtain a contradiction, suppose that $\weight(w) > \frac{1}{\alpha+1} \weight(v)$. In addition to $w$, $v$ must have $\alpha$ heavy children, each of which has a weight greater than or equal to $\weight(w)$. Hence
\[
\weight(v) \geq (1+\alpha) \cdot \weight(w) > (1+\alpha) \cdot \frac{1}{\alpha+1} \weight(v) = \weight(v) \; ,
\]
which is a contradiction.
\end{proof}
\autoref{lightdepth} holds for any heavy $\alpha$-tree decomposition obtained using a proper weight function on $T$. In the remaining part of the paper we will assume that the weight of a vertex is the number of leaves in the subtree rooted at $v$. See \autoref{fig:heavyalphatreedecomposition} for two different examples of heavy $\alpha$-tree decompositions.

We define $\lightheight_\alpha (T)$ to be the maximum light depth of a vertex in $T$, and remark that for $\alpha=0$, $\lightheight_\alpha(T)=\height(T)$. For a vertex $v$ in a compressed trie $T(S)$, we let $\lightstrings(v)$ denote the set of strings starting in one of the light edges leaving $v$. That is, $\lightstrings(v)$ is the union of the set of strings in the subtries $T_\ell(S)$ where $\ell$ is the first location on a light outgoing edge of $v$, i.e., $|\ell|=|v|+1$.

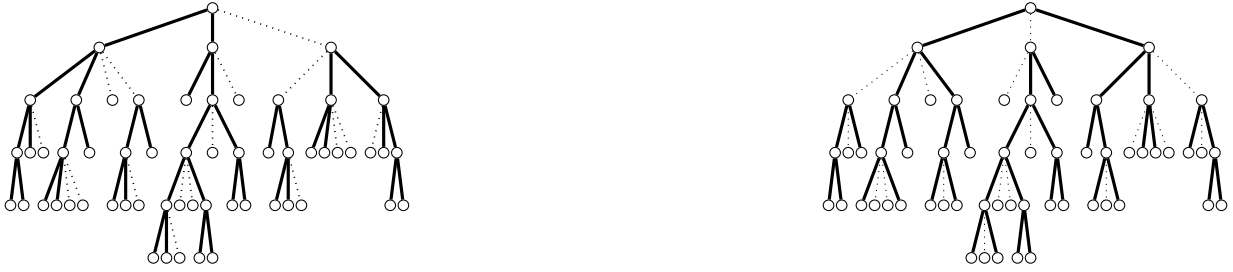
\begin{figure}\centering
\subfigure{
\begin{tikzpicture}[scale=0.35,
font=\scriptsize,
grow=down,
level 1/.style={level distance=1.5cm,sibling distance=4.5cm},
level 2/.style={level distance=2cm}
]
\begin{scope}[every child/.style={defnode},
every node/.style={defnode},
edge from parent/.style={heavyedge}
]
\def\ledge{edge from parent [lightedge]};
  \node (root) {}
  child [sibling distance=4.3cm]{
    node (a) {}
    child[sibling distance=1.75cm] {
      node (a1) {}
      child[sibling distance=0.5cm]{
        node (a2) {}
        child[sibling distance=0.5cm] { node (a3) {} }
        child[sibling distance=0.5cm] { node (a4) {} }
      }
      child [sibling distance=0.5cm]{ node (a5) {} }
      child [sibling distance=0.5cm]{ node  {} \ledge }
    }
    child [sibling distance=1.75cm] {
      node (b1) {}
      child [sibling distance=1cm] {
        node (b2) {}
        child [sibling distance=0.5cm]{ node  {} }
        child [sibling distance=0.5cm]{ node  {} }
        child [sibling distance=0.5cm]{ node  {} \ledge }
        child [sibling distance=0.5cm]{ node  {} \ledge }
      }
      child [sibling distance=1cm] { node {} }
    }
    child [sibling distance=1cm] { node {} \ledge }
    child [sibling distance=1cm] {
      node {}
      child [sibling distance=1cm] {
        node {}
        child [sibling distance=0.5cm]{ node  {} }
        child [sibling distance=0.5cm]{ node  {} }
        child [sibling distance=0.5cm]{ node  {} \ledge }
      }
      child [sibling distance=1cm] { node {} }
      \ledge
    }
   }
   child {
      node {}
      child [sibling distance=1cm] { node {} }
      child [sibling distance=1cm] {
        node {}
        child [sibling distance=1cm]{
          node  {}
          child [sibling distance=0.5cm] {
            node {}
            child [sibling distance=0.5cm] { node {} }
            child [sibling distance=0.5cm] { node {} }
            child [sibling distance=0.5cm] { node {} \ledge }
          }
          child [sibling distance=0.5cm] { node {} \ledge }
          child [sibling distance=0.5cm] { node {} \ledge }
          child [sibling distance=0.5cm] {
            node {}
            child [sibling distance=0.5cm] { node {} }
            child [sibling distance=0.5cm] { node {} }
          }
        }
        child [sibling distance=1cm]{ node  {} \ledge }
        child [sibling distance=1cm]{
          node  {}
          child [sibling distance=0.5cm] { node {} }
          child [sibling distance=0.5cm] { node {} }
        }
      }
      child [sibling distance=1cm] { node {} \ledge }
    }
   child {
     node {}
     child [sibling distance=2cm] {
       node {}
       child [sibling distance=0.75cm] { node {} }
       child [sibling distance=0.75cm] {
         node {}
         child [sibling distance=0.5cm] { node {} }
         child [sibling distance=0.5cm] { node {} }
         child [sibling distance=0.5cm] { node {} \ledge }
       }
       \ledge
     }
     child [sibling distance=2cm] {
       node {}
       child [sibling distance=0.5cm] { node {} }
       child [sibling distance=0.5cm] { node {} }
       child [sibling distance=0.5cm] { node {} \ledge }
       child [sibling distance=0.5cm] { node {} \ledge }
     }
     child [sibling distance=2cm] {
       node{}
       child [sibling distance=0.5cm] { node {} \ledge }
       child [sibling distance=0.5cm] { node {} }
       child [sibling distance=0.5cm] {
         node{}
         child [sibling distance=0.5cm] { node {} }
         child [sibling distance=0.5cm] { node {} }
       }
     }
     \ledge
   };
\end{scope}
\end{tikzpicture}
}
\hfill
\subfigure{
\begin{tikzpicture}[scale=0.35,
font=\scriptsize,
grow=down,
level 1/.style={level distance=1.5cm,sibling distance=4.5cm},
level 2/.style={level distance=2cm}
]
\begin{scope}[every child/.style={defnode},
every node/.style={defnode},
edge from parent/.style={heavyedge,very thick,draw}
]
\def\ledge{edge from parent [lightedge, thin]};
  \node (root) {}
  child [sibling distance=4.3cm]{
    node (a) {}
    child[sibling distance=1.75cm] {
      node (a1) {}
      child[sibling distance=0.5cm]{
        node (a2) {}
        child[sibling distance=0.5cm] { node (a3) {} }
        child[sibling distance=0.5cm] { node (a4) {} }
      }
      child [sibling distance=0.5cm]{ node {} \ledge }
      child [sibling distance=0.5cm]{ node {} }
      \ledge
    }
    child [sibling distance=1.75cm] {
      node (b1) {}
      child [sibling distance=1cm] {
        node (b2) {}
        child [sibling distance=0.5cm]{ node  {} }
        child [sibling distance=0.5cm]{ node  {} \ledge }
        child [sibling distance=0.5cm]{ node  {} \ledge }
        child [sibling distance=0.5cm]{ node  {} }
      }
      child [sibling distance=1cm] { node {} }
    }
    child [sibling distance=1cm] { node {} \ledge }
    child [sibling distance=1cm] {
      node {}
      child [sibling distance=1cm] {
        node {}
        child [sibling distance=0.5cm]{ node  {} }
        child [sibling distance=0.5cm]{ node  {} \ledge }
        child [sibling distance=0.5cm]{ node  {} }
      }
      child [sibling distance=1cm] { node {} }
    }
   }
   child {
      node {}
      child [sibling distance=1cm] { node {} \ledge }
      child [sibling distance=1cm] {
        node {}
        child [sibling distance=1cm]{
          node  {}
          child [sibling distance=0.5cm] {
            node {}
            child [sibling distance=0.5cm] { node {} }
            child [sibling distance=0.5cm] { node {} \ledge }
            child [sibling distance=0.5cm] { node {}  }
          }
          child [sibling distance=0.5cm] { node {} \ledge }
          child [sibling distance=0.5cm] { node {} \ledge }
          child [sibling distance=0.5cm] {
            node {}
            child [sibling distance=0.5cm] { node {} }
            child [sibling distance=0.5cm] { node {} }
          }
        }
        child [sibling distance=1cm]{ node  {} \ledge }
        child [sibling distance=1cm]{
          node  {}
          child [sibling distance=0.5cm] { node {} }
          child [sibling distance=0.5cm] { node {} }
        }
      }
      child [sibling distance=1cm] { node {} }
      \ledge
    }
   child {
     node {}
     child [sibling distance=2cm] {
       node {}
       child [sibling distance=0.75cm] { node {} }
       child [sibling distance=0.75cm] {
         node {}
         child [sibling distance=0.5cm] { node {} }
         child [sibling distance=0.5cm] { node {} \ledge }
         child [sibling distance=0.5cm] { node {} }
       }
     }
     child [sibling distance=2cm] {
       node {}
       child [sibling distance=0.5cm] { node {} \ledge }
       child [sibling distance=0.5cm] { node {} }
       child [sibling distance=0.5cm] { node {} }
       child [sibling distance=0.5cm] { node {} \ledge }
     }
     child [sibling distance=2cm] {
       node{}
       child [sibling distance=0.5cm] { node {} }
       child [sibling distance=0.5cm] { node {} \ledge }
       child [sibling distance=0.5cm] {
         node{}
         child [sibling distance=0.5cm] { node {} }
         child [sibling distance=0.5cm] { node {} }
       }
       \ledge
     }
   };
\end{scope}
\end{tikzpicture}
}
\caption{Two different heavy $\alpha$-tree decompositions with $\alpha=2$ of a tree with $n=38$ leaves. The maximum light depth is $3$ and $2$, respectively, in agreement with \autoref{lightdepth}.\label{fig:heavyalphatreedecomposition}}
\end{figure}

\subsection{Wildcard Trees}\label{wildcardtrees}
We introduce the \emph{$(\beta,k)$-wildcard tree}, denoted $T_\beta^k(C^\prime)$, where $1 \leq \beta < \sigma$ is a chosen parameter. This data structure stores a collection of strings $C^\prime \subset \Sigma^+$ in a compressed trie such that the search for a pattern $p$ with at most $k$ wildcards branches to at most $\beta$ locations in $T_\beta^k(C^\prime)$ when consuming a single wildcard of $p$. In particular for $\beta=1$, the search for $p$ never branches and the search time becomes linear in the length of $p$. For a vertex $v$, we define the wildcard height of $v$ to be the number of wildcards on the path from $v$ to the root. Intuitively, given a wildcard tree that supports $i$ wildcards, support for an extra wildcard is added by joining a new tree to each vertex $v$ with wildcard height $i$ by an edge labeled $\wildcard$. This tree is searched if a wildcard is consumed in $v$. Formally, $T_\beta^k(C^\prime)$ is built recursively as follows.

\begin{quote}\textbf{Construction of $T_\beta^i(S)$}: 
Produce a heavy $(\beta-1)$-tree decomposition of $T(S)$, then for each internal vertex $v \in T(S)$ join $v$ to the root of $T_\beta^{i-1}(\suff_2(\lightstrings(v))$ by an edge labeled $\wildcard$. Let $T_\beta^0(S) = T(S)$.
\end{quote}

\noindent The construction is illustrated in \autoref{fig:wildcardtreeconstructionabstract}. Since a leaf $\ell$ in a compressed trie $T(S)$ is obtained as the suffix of a string $x \in C^\prime$, we assume that $\ell$ inherits the label of $x$ in case the strings in $C^\prime$ are labeled. For example, when $C^\prime$ denotes the suffixes of $t$, we will label each suffix in $C^\prime$ with its start position in $t$. This immediately provides us with a $k$-bounded wildcard index. \autoref{fig:wildcardsearchtreeconstruction} shows some concrete examples of the construction of $T_\beta^k(C^\prime)$ when $C^\prime$ is a set of labeled suffixes. 

\begin{figure}[H]\centering \footnotesize
\begin{tikzpicture}[scale=1]
\draw (2.5,4) coordinate(r);
\draw (0,0) coordinate(left);
\draw (5,0) coordinate(right);

\draw (6.2,2.5) coordinate(r2);
\draw (5.2,1) coordinate(left2);
\draw (7.2,1) coordinate(right2);

\draw (6.3,1.8) coordinate(v2);

\draw (10,2.3) coordinate(r3);
\draw (9.5,1.3) coordinate(left3);
\draw (10.5,1.3) coordinate(right3);

\draw (6,3.8) coordinate(r4);
\draw (5,1.8) coordinate(left4);
\draw (7,1.8) coordinate(right4);

\draw (6.1,3.3) coordinate(v4);

\draw (6,0.9) coordinate(r5);
\draw (5.5,0) coordinate(left5);
\draw (6.5,0) coordinate(right5);

\draw (6.08,0.5) coordinate(v5);

\draw (10.2,3.5) coordinate(r6);
\draw (9.7,2.5) coordinate(left6);
\draw (10.7,2.5) coordinate(right6);

\draw (10.2,1) coordinate(r7);
\draw (9.7,0) coordinate(left7);
\draw (10.7,0) coordinate(right7);

\draw (11,1.8) coordinate(r8);
\draw (10.5,0.8) coordinate(left8);
\draw (11.5,0.8) coordinate(right8);

\draw (11.1,3.1) coordinate(r9);
\draw (10.6,2.1) coordinate(left9);
\draw (11.6,2.1) coordinate(right9);

\draw ($(left)!0.5!(right)$) coordinate(leaf);
\draw ($(r)!0.6!(leaf)$) coordinate(v);
\foreach \x/\y in {-45/1,-35/2,-25/3,-15/4,-5/5,5/6,15/7,25/8,35/9,45/10}
\draw ($(v)!20pt!\x:(leaf)$) coordinate(child\y);
\draw[fill=gray,opacity=0.4] (r) -- (left) -- (right) -- cycle;
\draw[fill=gray,opacity=0.04] (r4) -- (left4) -- (right4) -- cycle;
\draw[fill=gray,opacity=0.04] (r5) -- (left5) -- (right5) -- cycle;
\draw[fill=gray,opacity=0.1] (r3) -- (left3) -- (right3) -- cycle;
\draw[fill=gray,opacity=0.1] (r6) -- (left6) -- (right6) -- cycle;
\draw[fill=gray,opacity=0.1] (r7) -- (left7) -- (right7) -- cycle;
\draw[fill=gray,opacity=0.1] (r8) -- (left8) -- (right8) -- cycle;
\draw[fill=gray,opacity=0.1] (r9) -- (left9) -- (right9) -- cycle;
\draw[fill=gray,opacity=0.5] (r2) -- (left2) -- (right2) -- cycle;

\draw (r2) node[defnode,fill=white,label={above:$T_\beta^{k-1}(\suff_2(\lightstrings(v)))$}](r2node) {};

\draw (r3) node[defnode,draw=black!20,fill=white](r3node) {};
\draw (r4) node[defnode,draw=black!20,fill=white](r4node) {};
\draw (r5) node[defnode,draw=black!20,fill=white](r5node) {};
\draw (r6) node[defnode,draw=black!20,fill=white](r6node) {};
\draw (r7) node[defnode,draw=black!20,fill=white](r7node) {};
\draw (r8) node[defnode,draw=black!20,fill=white](r8node) {};
\draw (r9) node[defnode,draw=black!20,fill=white](r9node) {};

\draw (v2) node[defnode,fill=white](v2node) {};

\draw (v4) node[defnode,fill=white,opacity=0.2](v4node) {};

\draw (v5) node[defnode,fill=white,opacity=0.2](v5node) {};

\draw (r3node) edge[bend right=15,opacity=0.2] node[midway,above] {$\wildcard$} ($(r3node)!0.3!(r2node)$);

\draw (r6node) edge[bend right=15,opacity=0.2] node[midway,above] {$\wildcard$} ($(r6node)!0.3!(r4node)$);
\draw (r7node) edge[bend right=15,opacity=0.2] node[near end,above] {$\wildcard$} ($(r7node)!0.3!(r5node)$);
\draw (r8node) edge[bend right=15,opacity=0.2] node[near end,above] {$\wildcard$} ($(r8node)!0.4!(r2node)$);
\draw (r9node) edge[bend right=15,opacity=0.2] node[near end,above] {$\wildcard$} ($(r9node)!0.4!(v4node)$);

\draw (v4node) edge[bend left=10,opacity=0.2] node[near end,label={above:$\wildcard$}] {} ($(v4node)!0.3!(r6node)$);

\draw (v2node) edge[bend left=10] node[near end,label={above:$\wildcard$}] {} ($(v2node)!0.3!(r3node)$);

\draw (v5node) edge[bend left=10,opacity=0.2] node[near end,label={above:$\wildcard$}] {} ($(v5node)!0.4!(r7node)$);

\begin{scope}[opacity=0.3]
\draw ($(r)!0.2!(v)$) node[defnode,fill=white] {} edge[bend left=20] node[midway,above] {$\wildcard$} (r4node);
\draw ($(v)!0.3!(right)$) node[defnode,fill=white] {} edge[bend left=20] node[midway,above] {$\wildcard$} (r5node);
\end{scope}

\draw (v)
node[defnode,label={above:$v$}] (vnode) {}
\foreach \x in {child1,child2,child4}
{ edge [heavyedge] (\x) }
\foreach \x in {child6,child7,child9}
{ edge [lightedge] (\x) }
;
\draw[densely dotted] ($(v)!0.8!(child2)$) -- ($(v)!0.8!(child4)$);
\draw[densely dotted] ($(v)!0.8!(child7)$) -- ($(v)!0.8!(child9)$);

\draw (vnode) edge[bend left=20] node[midway,label={above:\footnotesize $\wildcard$}] {} (r2node);

\draw[decorate,decoration={brace,raise=2pt,amplitude=3pt}] (child4) -- (child1) node[midway,below=3pt,xshift=-4pt] {\scriptsize $\beta -1$};
\draw[decorate,decoration={brace,raise=2pt,amplitude=3pt}] (child9) -- (child6) node[midway,below=2pt,xshift=25pt] {\scriptsize $\lightstrings(v)$};

\draw[decorate,decoration={brace,raise=5pt,amplitude=3pt}] (right) -- (left) node[midway,below=6pt] {\scriptsize $T_\beta^0(C^\prime)$};

\draw[decorate,decoration={brace,raise=5pt,amplitude=3pt}] (7.2,-0.5) -- (0,-0.5) node[midway,below=6pt] {\scriptsize $T_\beta^1(C^\prime)$};

\draw[decorate,decoration={brace,raise=5pt,amplitude=3pt}] (11.6,-1) -- (0,-1) node[midway,below=6pt] {\scriptsize $T_\beta^k(C^\prime)$};


\draw[line width=2,loosely dotted, draw=black!60] (7.5,2.2) -- (8.9,2.2);
\end{tikzpicture}
\caption{Illustrating of the recursive construction of the wildcard tree $T_\beta^k(C^\prime)$. The final tree consists of $k$ layers of compressed tries joined by edges labeled $\wildcard$.\label{fig:wildcardtreeconstructionabstract}
}
\end{figure}
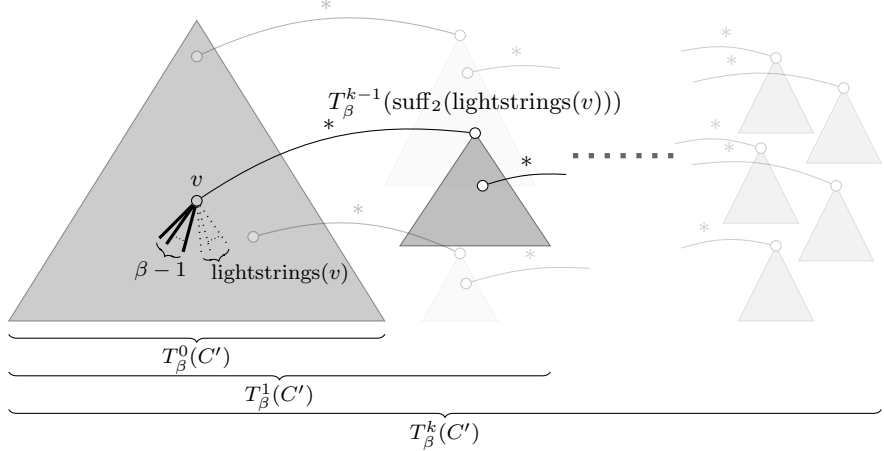

\begin{figure}[p]\centering
\subfigure[$T_1^2(C)$.]{
\begin{tikzpicture}[scale=0.45,
font=\scriptsize,
grow=down,
level distance=3cm,
level 1/.style={sibling distance=9.5cm},
level 2/.style={sibling distance=5cm},
level 3/.style={sibling distance=1.6cm}]
\begin{scope}[every child/.style={defnode},
every node/.style={defnode,outer sep=0.6pt},
lab/.style={draw=none,rectangle,sloped,allow upside down,fill=white,font=\scriptsize\ttfamily},
lvl1/.style={},
lvl2/.style={double},
lvl3/.style={double,label={[font=\tiny]center:\tikz \draw circle (0.02);}}
]
\def\pedge{edge from parent [lightedge] node[lab]};
  \node[lvl1]{}
  child[sibling distance=3.5cm] {
    node{}
    child[sibling distance=2.4cm] {
      node {}
      child[sibling distance=1cm] {node[label=below:2]{} \pedge {nas\$}}
      child[sibling distance=1cm] {node[label=below:4]{} \pedge {s\$}}
      child[sibling distance=1cm,every child node/.style={defnode,lvl2}] {
        node{}
        child[sibling distance=1cm] {node[label=below:2]{} \pedge {as\$}}
        child[sibling distance=1cm] {node[label=below:4]{} \pedge {\$}}
        child[sibling distance=1cm,every child node/.style={lvl3}] {
          node{}
          child {node[label=below:2]{} \pedge {s\$}}
          \pedge {$\wildcard$}
        }
        \pedge {$\wildcard$}
      }
      \pedge {na}
    }
    child[sibling distance=2.1cm] {node[label=below:6]{} \pedge {s\$}}
    child[sibling distance=2.1cm,every child node/.style={lvl2}] {
      node{}
      child[sibling distance=1cm] {
        node {}
        child[sibling distance=1cm] {node[label=below:2]{} \pedge {nas\$}}
        child[sibling distance=1cm] {node[label=below:4]{} \pedge {s\$}}
        child[sibling distance=1cm,every child node/.style={lvl3}] {
          node{}
          child[sibling distance=1cm] {node[label={below:2}]{} \pedge {as\$}}
          child[sibling distance=1cm] {node[label=below:4]{} \pedge {\$}}
          \pedge {$\wildcard$}
        }
        \pedge {a}
      }
      child[sibling distance=1cm] {node[label=below:6]{} \pedge {\$}}
      child[sibling distance=1cm,every child node/.style={lvl3}] {
        node{}
        child[sibling distance=1cm] {node[label=below:2]{} \pedge {nas\$}}
        child[sibling distance=1cm] {node[label=below:4]{} \pedge {s\$}}
        \pedge {$\wildcard$}
      }
      \pedge {$\wildcard$}
    }
    \pedge {a}
  }
  child[sibling distance=4cm] {
    node[label=below:1]{} \pedge {bananas\$}
  }
  child[sibling distance=9cm] {
        node{}
        child[sibling distance=1cm] {node[label=below:3]{} \pedge {nas\$}}
        child[sibling distance=1cm] {node[label=below:5]{} \pedge {s\$}}
        child[sibling distance=1cm,every child node/.style={lvl2}] {
          node{}
          child[sibling distance=1cm] {node[label=below:3]{} \pedge {as\$}}
          child[sibling distance=1cm] {node[label=below:5]{} \pedge {\$}}
          child[sibling distance=1cm,every child node/.style={lvl3}] {
            node{}
            child {node[label=below:3]{} \pedge {s\$}}
            \pedge {$\wildcard$}
          }
          \pedge {$\wildcard$}
        }
       \pedge {na}
    }
    child[missing,sibling distance=0cm]
    child[sibling distance=1cm] {
      node[label=below:7]{} \pedge {s\$}
    }
    child[sibling distance=3.2cm,every child node/.style={lvl2}] {
      node{}
      child[sibling distance=3.3cm] {
        node{}
        child[sibling distance=1.5cm] {
          node {}
          child[sibling distance=1cm] {node[label=below:1]{} \pedge {nas\$}}
          child[sibling distance=1cm] {node[label=below:3]{} \pedge {s\$}}
          child[sibling distance=1cm,every child node/.style={lvl3}] {
            node{}        
            child[sibling distance=1cm] {node[label=below:1]{} \pedge {as\$}}
            child[sibling distance=1cm] {node[label=below:3]{} \pedge {\$}}
            \pedge {$\wildcard$}}
          child[missing]
          \pedge {na}
        }
        child[sibling distance=1.5cm] {node[label=below:5]{} \pedge {s\$}}
        child[sibling distance=1.5cm,every child node/.style={lvl3}] {
          node{}
          child[sibling distance=1cm] {
            node {}
            child[sibling distance=1cm] {node[label=below:1]{} \pedge {nas\$}}
            child[sibling distance=1cm] {node[label=below:3]{} \pedge {s\$}}
            \pedge {a}
          }
          child[sibling distance=1cm] {node[label=below:5]{} \pedge {\$}}
          \pedge {$\wildcard$}
        }
        \pedge {a}
      }
      child[sibling distance=2.5cm] {
          node {}
          child[sibling distance=1.5cm] {node[label=below:2]{} \pedge {nas\$}}
          child[sibling distance=1.5cm] {node[label=below:4]{} \pedge {s\$}}
          child[sibling distance=1.5cm,every child node/.style={lvl3}] {
            node{}
            child[sibling distance=1cm] {node[label=below:2]{} \pedge {as\$}}
            child[sibling distance=1cm] {node[label=below:4]{} \pedge {\$}}
            \pedge {$\wildcard$}
          }
          \pedge {na}
        }
      child[sibling distance=1cm] {node[label=below:6]{} \pedge {s\$}}
      child[sibling distance=1.5cm] {node[label=below:7]{} \pedge {\$}}
      child[sibling distance=2cm,every child node/.style={lvl3}] {
        node{}
        child[sibling distance=2cm] {
          node{}
          child[sibling distance=1cm] {node[label=below:2]{} \pedge {nas\$}}
          child[sibling distance=1cm] {node[label=below:4]{} \pedge {s\$}}
          \pedge {a}
        }
        child[sibling distance=2cm] {
          node{}
          child[sibling distance=1cm] {node[label=below:1]{} \pedge {nas\$}}
          child[sibling distance=1cm] {node[label=below:3]{} \pedge {s\$}}
          \pedge {na}
        }
        child[sibling distance=1cm] {node[label=below:5]{} \pedge {s\$}}
        child[sibling distance=1cm] {node[label=below:6]{} \pedge {\$}}
        \pedge {$\wildcard$}}
      \pedge {$\wildcard$}
    }
;
\end{scope}
\end{tikzpicture}
}
\subfigure[$T^2_2(C)$.]{
\begin{tikzpicture}[scale=0.45,
font=\scriptsize,
grow=down,
level distance=3cm,
level 1/.style={sibling distance=9.5cm},
level 2/.style={sibling distance=5cm},
level 3/.style={sibling distance=1.6cm},
lvl1/.style={},
lvl2/.style={double},
lvl3/.style={double,label={[font=\tiny]center:\tikz \draw circle (0.02);}},
baseline=(current bounding box.north)]
\draw[use as bounding box] (-7,0.5) (5,-13);
\begin{scope}[every child/.style={defnode},
every node/.style={defnode},
lab/.style={draw=none,rectangle,sloped,allow upside down,fill=white,font=\scriptsize\ttfamily}
]
\def\ledge{edge from parent [lightedge] node[lab]};
\def\hedge{edge from parent [heavyedge] node[lab]};
  \node{}
  child[sibling distance=2.5cm] {
    node{}
    child[sibling distance=1cm] {
      node {}
      child[sibling distance=1cm] {node[label=below:2]{} \hedge {nas\$}}
      child[sibling distance=1cm] {node[label=below:4]{} \ledge {s\$}}
      child[sibling distance=1cm,every child node/.style={lvl2}] {
        node {}
        child[sibling distance=1cm] {node[label=below:4] {} \hedge {\$}}
        \ledge {$\wildcard$}
      }
      \hedge {na}
    }
    child[sibling distance=1cm] {node[label=below:6]{} \ledge {s\$}}
    child[sibling distance=1cm,every child node/.style={lvl2}] {
      node {}
      child[sibling distance=1cm] {node[label=below:6] {} \hedge {\$}}
      \ledge {$\wildcard$}
    }
    \hedge {a}
  }
  child[sibling distance=2.5cm] {
    node[label=below:1]{} \ledge {bananas\$}
  }
  child[sibling distance=2cm] {
    node{}
    child[sibling distance=1cm] {node[label=below:3]{} \hedge {nas\$}}
    child[sibling distance=1cm] {node[label=below:5]{} \ledge {s\$}}
    child[sibling distance=1cm,every child node/.style={lvl2}] {
      node {}
      child[sibling distance=1cm] {node[label=below:5] {} \hedge {\$}}
      \ledge {$\wildcard$}
    }
    \ledge {na}
  }
  child[sibling distance=2.5cm] {
    node[label=below:7]{} \ledge {s\$}
  }
  child[sibling distance=2.5cm,every child node/.style={lvl2}] {
    node {}
    child [sibling distance=2cm]{
      node {}
      child [sibling distance=1cm]{
        node{}
        child[sibling distance=1cm] {node[label=below:1] {} \hedge {nas\$}}
        child[sibling distance=1cm] {node[label=below:3] {} \ledge {s\$}}
        child[sibling distance=1cm,every child node/.style={lvl3}] {
          node {}
          child[sibling distance=1cm,level distance=1.5cm] {node[label=below:3] {} \hedge {\$}}
          \ledge {$\wildcard$}
        }
        \hedge {na}
      }
      child[sibling distance=1cm] {node[label=below:5] {} \ledge {s\$}}
      child[sibling distance=1cm,every child node/.style={lvl3}] {
        node {}
        child[sibling distance=1cm] {node[label=below:5] {} \hedge {\$}}
        \ledge {$\wildcard$}
      }
      \hedge {a}
    }
    child[sibling distance=2cm] {node[label=below:7] {} \ledge {\$}}
    \ledge {$\wildcard$}
  }
;
\end{scope}
\end{tikzpicture}
}
\hfill
\subfigure[$T_3^2(C)$.]{
\begin{tikzpicture}[scale=0.45,
font=\scriptsize,
grow=down,
level distance=3cm,
level 1/.style={sibling distance=9.5cm},
level 2/.style={sibling distance=5cm},
level 3/.style={sibling distance=1.6cm},
lvl1/.style={},
lvl2/.style={double},
lvl3/.style={double,label={[font=\tiny]center:\tikz \draw circle (0.02);}},
baseline=(current bounding box.north)
]
\draw[use as bounding box] (-6,0.5) (6,-13);
\begin{scope}[every child/.style={defnode},
every node/.style={defnode},
lab/.style={draw=none,rectangle,sloped,allow upside down,fill=white,font=\scriptsize\ttfamily}
]
\def\ledge{edge from parent [lightedge] node[lab]};
\def\hedge{edge from parent [heavyedge] node[lab]};
  \node{}
  child[sibling distance=2.5cm] {
    node{}
    child[sibling distance=1cm] {
      node {}
      child[sibling distance=1cm] {node[label=below:2]{} \hedge {nas\$}}
      child[sibling distance=1cm] {node[label=below:4]{} \hedge {s\$}}
      \hedge {na}
    }
    child[sibling distance=1cm] {node[label=below:6]{} \hedge {s\$}}
    \hedge {a}
  }
  child[sibling distance=2.5cm] {
    node[label=below:1]{} \ledge {bananas\$}
  }
  child[sibling distance=2cm] {
    node{}
    child[sibling distance=1cm] {node[label=below:3]{} \hedge {nas\$}}
    child[sibling distance=1cm] {node[label=below:5]{} \hedge {s\$}}
    \hedge {na}
  }
  child[sibling distance=2.5cm] {
    node[label=below:7]{} \ledge {s\$}
  }
  child[sibling distance=2.5cm,every child node/.style={lvl2}] {
    node {}
    child[sibling distance=1cm] {node[label=below:1] {} \hedge {ananas\$}}
    child[sibling distance=1cm] {node[label=below:7] {} \hedge {\$}}
    \ledge {$\wildcard$}
  }
;
\end{scope}
\end{tikzpicture}
}
\caption{Showing $T_\beta^k(C)$ for $\beta \in \{1,2,3\}$, $k=2$ and $C=\suff(\texttt{bananas\$})$. The recursion levels 0, 1, 2 in the construction are indicated by increasing growth rings in the vertices. All edges in $T_1^2(C)$ are light, since the construction is based on a heavy $\alpha$-tree decomposition with $\alpha=\beta-1=0$. Leaves are labeled with the start position of their corresponding suffix in $t$. \label{fig:wildcardsearchtreeconstruction}}
\end{figure}
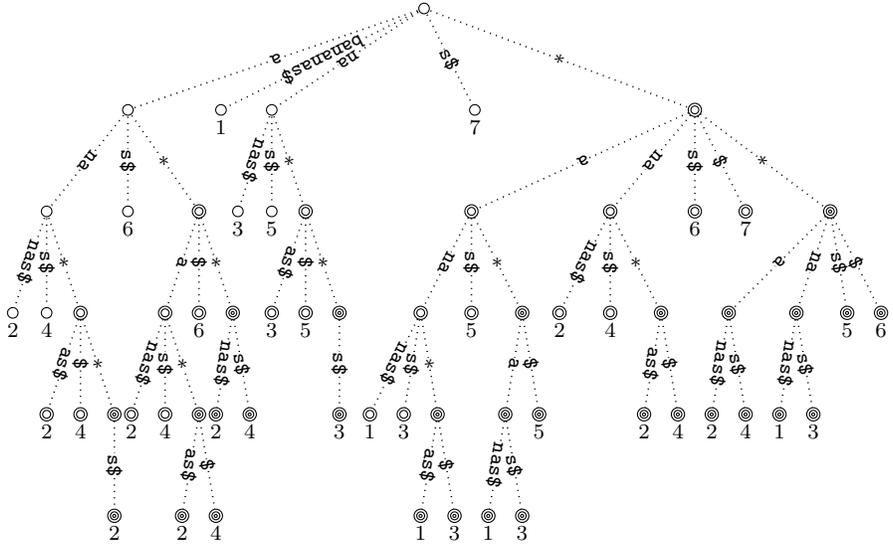
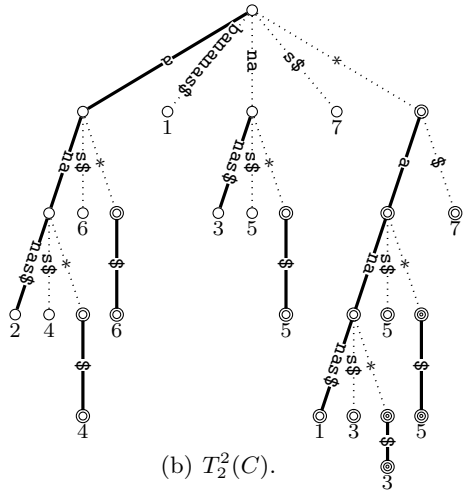
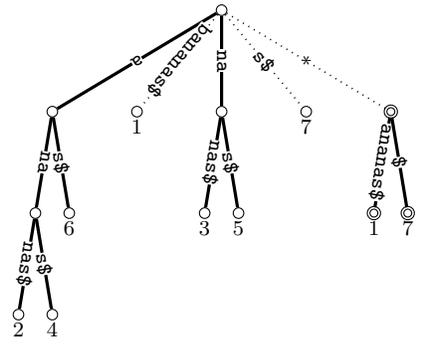


\subsection{Wildcard Tree Index}\label{wildcardtreeindex}
Given a collection $C^\prime$ of strings and a pattern $p$, we can identify the strings of $C^\prime$ having a prefix matching $p$ by constructing $T_\beta^k(C^\prime)$. Searching $T_\beta^k(C^\prime)$ is similar to the suffix tree search, except when consuming a wildcard character of $p$ in an explicit vertex $v \in T_\beta^k(C^\prime)$ with more than $\beta$ children. In that case the search branches to the root of the wildcard tree joined to $v$ and to the first location on the $\beta-1$ heavy edges of $v$, effectively letting the wildcard match the first character on all edges from $v$. Consequently, the search for $p$ branches to a total of at most $\sum_{i=0}^j \beta^i = O(\beta^j)$ locations, each of which requires $O(m)$ time, resulting in a query time $O(\beta^j m + occ)$. For $\beta=1$ the query time is $O(m+j+occ)$.

\begin{lemma}\label{wildcard-search-tree}
For any integer $1 \leq \beta < \sigma$, the wildcard tree $T_\beta^k(C^\prime)$ has query time $O(\beta^jm+j+occ)$. The wildcard tree stores $O(|C^\prime| H^k)$ strings, where $H$ is an upper bound on the light height of all compressed tries $T(S)$ satisfying $S \subseteq \suff_d(C^\prime)$ for some integer $d$.
\end{lemma}

\begin{proof}
We prove that the total number of strings (leaves) in $T_\beta^i(S)$, denoted $|T_\beta^i(S)|$, is at most $|S|\sum_{j=0}^i H^j = O(|S|H^i)$. The proof is by induction on $i$. The base case $i=0$ holds, since $T_\beta^0(S)=T(S)$ contains $|S| = |S| \sum_{j=0}^0 H^j$ strings. For the inductive step, assume that $|T_\beta^i(S)| \leq |S| \sum_{j=0}^i H^j$. Let $S_v=\suff_2(\lightstrings(v))$ for a vertex $v \in T(S)$. From the construction we have that the number of strings in $T_\beta^{i+1}(S)$ is the number of strings in $T(S)$ plus the number of strings in the wildcard trees joined to the vertices of $T(S)$. That is,
\[
\bigl| T_\beta^{i+1}(S) \bigr| ~=~ \bigl|S \bigr| + \sum_{v \in T(S)} \bigl| T_\beta^i(S_v) \bigr| ~ \stackrel{IH}{\leq } ~ \bigl| S \bigr| + \sum_{v \in T(S)} \bigl| S_v \bigr| \sum_{j=0}^i H^j\; .
\]
The string sets $S_v$ consist of suffixes of strings in $S$. Consider a string $x \in S$, i.e., a leaf in $T(S)$. The number of times a suffix of $x$ appears in a set $S_v$ is equal to the light depth of $x$ in $T(S)$. $S$ is also a set of suffixes of $C^\prime$, and hence $H$ is an upper bound on the maximum light depth of $T(S)$. This establishes that
\(
\sum_{v \in T(S)} |S_v| ~\leq~ |S| H \; ,
\)
thus showing that $|T_\beta^{i+1}(S)| \leq |S| + |S|H \sum_{j=0}^i H^j = |S| \sum_{j=0}^{i+1} H^j$.
\end{proof}
\noindent Constructing the wildcard tree $T_\beta^k(C)$, where $C = \suff(t)$, we obtain a wildcard index with the following properties.
\begin{lemma}\label{wildcard-tree-index}
Let $t$ be a string of length $n$ from an alphabet of size $\sigma$. For $2 \leq \beta < \sigma$ there is a $k$-bounded wildcard index for $t$ using $O \bigl( n \log_\beta^{k} n \bigr )$ space. The index can report the occurrences of a pattern with $m$ characters and $j \leq k$ wildcards in time $O \left (\beta^j m + occ \right )$.
\end{lemma}
\begin{proof}
The query time follows from \autoref{wildcard-search-tree}. Since $T_\beta^k(C)$ is a compressed trie, and because each edge label is a substring of $t$, the space needed to store $T_\beta^k(C)$ is upper bounded by the number of strings it contains which by \autoref{wildcard-search-tree} is $O(n H^k)$. It follows from \autoref{lightdepth} that $H=\log_\beta n$ is an upper bound on the light height of all compressed tries $T(S)$, since they each contain at most $n$ vertices. Consequently, the space needed to store the index is $O(n \log_\beta^k n)$.
\end{proof}

\subsection{Wildcard Tree Index Using the LCP Data Structure}\label{wildcardtreeindexlcp}
The wildcard index of \autoref{wildcard-tree-index} reduces the branching factor of the suffix tree search from $\sigma$ to $\beta$, but still has the drawback that the search for a subpattern $p_i$ from a location $\ell \in T_\beta^k(C)$ takes $O(|p_i|)$ time. This can be addressed by combining the index with the LCP data structure as in Cole~et~al.~\cite{Cole2004}. In that way, the search for a subpattern can be done in time $O(\log\log n)$. The index is obtained by modifying the construction of $T_\beta^i(S)$ such that each $T(S)$ is added to the LCP data structure prior to joining the $(\beta,i-1)$-wildcard trees to the vertices of $T(S)$. For all $T(S)$ except the final $T(S)=T_\beta^0(S)$, support for unrooted LCP queries in time $O(\log\log n)$ is added using additional $O(|S|\log|S|)$ space. For the final $T(S)$, searched when all $k$ wildcards have been matched, we only need support for rooted queries. Upon receiving the query pattern $p=p_1 \wildcard p_2 \wildcard \ldots \wildcard p_k$, each $p_i$ is preprocessed in time $O(|p_i|)$ to support LCP queries for any suffix of $p_i$. The search for $p$ proceeds as described for the normal wildcard tree, except now rooted and unrooted LCP queries are used to search for suffixes of $p_0,p_1,\ldots,p_k$.

In the search for $p$, a total of at most $\sum_{i=0}^j \beta^i = O(\beta^j)$ LCP queries, each taking time $O(\log\log n)$, are performed. Preprocessing $p_0,p_1,\ldots,p_j$ takes $\sum_{i=0}^j|p_i| = m$ time, so the query time is $O(m+\beta^j \log\log n + occ)$. The space needed to store the index is $O(n \log_\beta^k n)$ for $T_\beta^k(C)$ plus the space needed to store the LCP data structure. 

Adding support for rooted LCP queries requires linear space in the total size of the compressed tries, i.e., $O(n \log_\beta^k n)$. Let $T(S_0),T(S_1),\ldots,T(S_q)$ denote the compressed tries with support for unrooted LCP queries. Since each $S_i$ contains at most $n$ strings and $\sum_{i=0}^q |S_i| = |T_\beta^{k-1}(C)|$, by \autoref{lcp-ds}, the additional space required to support unrooted LCP queries is
\begin{footnotesize}
\[
O \bigl( \sum_{i=0}^q |S_i| \log |S_i| \bigr) = O \bigl( \log n \sum_{i=0}^q |S_i| \bigr)= O \left( \log n |T_\beta^{k-1} (C^\prime)| \right) = O \bigl( n\log(n) \log_\beta^{k-1} n \bigr ) \; ,
\]
\end{footnotesize}which is an upper bound on the total space required to store the wildcard index. This concludes the proof of \autoref{wst-index}. The $k$-bounded wildcard index described by Cole~et~al.~\cite{Cole2004} is obtained as a special case of \autoref{wst-index}.

\begin{corollary}[Cole~et~al.~\cite{Cole2004}]
Let $t$ be a string of length $n$ from an alphabet of size $\sigma$. There is a $k$-bounded wildcard index for $t$ using $O(n\log^k n)$ space. The index can report the occurrences of a pattern with $m$ characters and $j \leq k$ wildcards in time $O(m+2^j\log\log n + occ)$.
\end{corollary}

\section{A $k$-Bounded Wildcard Index with Linear Query Time}\label{sec:optimaltimeindex}
%
Consider the $k$-bounded wildcard index obtained by creating the wildcard tree $T_1^k(\suff(t))$ for $t$. This index has linear query time, and we can show that the space usage depends of the height of the suffix tree. 
\begin{lemma}\label{opt-time-height}
  Let $t$ be a string of length $n$ from an alphabet of size $\sigma$. There is a $k$-bounded wildcard index for $t$ using $O(nh^k)$ space, where $h$ is the height of the suffix tree for $t$. The index can report the occurrences of a pattern with $m$ characters and $j$ wildcards in time $O(m+ j + occ)$.
\end{lemma}
\begin{proof}
Since $\suff(t)$ is closed under the suffix operation, the height of $T(\suff(t))$ is an upper bound on the height of all compressed tries $T(S)$ satisfying $S \subseteq \suff_d(\suff(t))$ for some $d$. For $\beta=1$, the light height of $T(S)$ is equal to the height of $T(S)$, so $H=h=\height(T(\suff(t)))$ can be used as an upper bound of the light height in \autoref{wildcard-search-tree}, and consequently the space needed to store $T_1^k(\suff(t))$ is $O(nh^k)$.
\end{proof}
In the worst case the height of the suffix tree is close to $n$, but combining the index with another wildcard index yields a useful black box reduction. The idea is to query the first index if the pattern is short, and the second index if the pattern is long.

\begin{lemma}\label{cuttop}
Let $F \geq m$ and let $G$ be independent of $m$ and $j$. Given a wildcard index $\mathcal{A}$ with query time $O(F+G+occ)$ and space usage $S$, there is a $k$-bounded wildcard index $\mathcal{B}$ with query time $O(F+j+occ)$ and taking space $O(n \min(G,h)^k + S)$, where $h$ is the height of the suffix tree for $t$.
\end{lemma}

\begin{proof}
The wildcard index $\mathcal{B}$ consists of $\mathcal{A}$ as well as a special wildcard index $T_1^k(\pref_G(\suff(t)))$ $\mathcal{C}$, which is a wildcard tree with $\beta=1$ over the set of all substrings of $t$ of length $G$. $G$ can be used as an upper bound for the light height in \autoref{wildcard-search-tree}, so the space required to store $\mathcal{C}$ is $O(n \min(G,h)^k)$ by using \autoref{opt-time-height} if $G>h$. A query on $\mathcal{B}$ results in a query on either $\mathcal{A}$ or $\mathcal{C}$. In case $G < F+j$, we query $\mathcal{A}$ and the query time will be $O(F+G+occ) = O(F+j+occ)$. In case $G \geq F+j$, we query $\mathcal{C}$ with query time $O(m+j+occ) = O(F+j+occ)$. In any case the query time of $\mathcal{B}$ is $O(F+j+occ)$.
\end{proof}
Applying \autoref{cuttop} with $F=m$ and $G=\sigma^k \log\log n$ on the unbounded wildcard index from \autoref{optimalspaceindex} yields a new $k$-bounded wildcard index with linear query time using space $O(\sigma^{k^2} n \log^k \log n)$. This concludes the proof of \autoref{optimaltimeindex}.

\section{Variable Length Gaps}\label{vl-gaps}\label{sec:vlg}
We now consider the \emph{string indexing for patterns with variable length gaps problem}. By only changing the search procedure, this problem can be solved using the previously described bounded and unbounded wildcard indexes.

The string indexing for patterns with variable length gaps problem is to build an index for a string $t$ that can efficiently report the occurrences of a query pattern $p$ of the form
\[
p = p_0\; \gap{a_1}{b_1}\; p_1\; \gap{a_2}{b_2}\; \ldots\; \gap{a_j}{b_j}\; p_j \; .
\]
The query pattern consists of $j+1$ strings $p_0,p_1,\ldots,p_j \in \Sigma^*$ interleaved by $j$ \emph{variable length gaps} $\gap{a_i}{b_i}$, $i=1,\ldots,j$, where $a_i$ and $b_i$ are positive integers such that $a_i \leq b_i$. Intuitively, a variable length gap $\gap{a_i}{b_i}$ matches an arbitrary string over $\Sigma$ of length between $a_i$ and $b_i$, both inclusive.

\begin{example}\label{ex:vlg}
Consider the string $t$ and pattern $p$ over the alphabet $\Sigma = \{\texttt{a},\texttt{b},\texttt{c},\texttt{d} \}$.
\begin{align*}
t &= \texttt{acbccbacccddabdaabcdccbccdaa} \\
p &= \texttt{b} \gap{0}{4} \texttt{cc} \gap{3}{5} \texttt{d}
\end{align*}
The string $t$ contains five occurrences of the query pattern $p$ as shown in \autoref{fig:vlgexample}.
\end{example}

\begin{figure}[b]\centering
\begin{tikzpicture}[scale=0.4,yscale=0.8]\ttfamily
\draw (0,0) \foreach \x / \y in {
1/a,
2/c,
3/b,
4/c,
5/c,
6/b,
7/a,
8/c,
9/c,
10/c,
11/d,
12/d,
13/a,
14/b,
15/d,
16/a,
17/a,
18/b,
19/c,
20/d,
21/c,
22/c,
23/b,
24/c,
25/c,
26/d,
27/a,
28/a
} {
   (\x,0) node[anchor=base]{\scriptsize \textnormal{\x}}
   (\x,-1) node[anchor=base]{\texttt{\textbf{\y}}}
 }
;
\draw (0.5,0.7) -- (28.5,0.7);
\draw (0.5,-0.2) -- (28.5,-0.2);
\draw (3,-2) node[anchor=base]{b};
\draw (4,-2) node[anchor=base]{c};
\draw (5,-2) node[anchor=base]{c};
\draw (6,-2) node[anchor=base]{$\wildcard$};
\draw (7,-2) node[anchor=base]{$\wildcard$};
\draw (8,-2) node[anchor=base]{$\wildcard$};
\draw (9,-2) node[anchor=base]{$\wildcard$};
\draw (10,-2) node[anchor=base]{$\wildcard$};
\draw (11,-2) node[anchor=base]{d};
\draw (3,-3) node[anchor=base]{b};
\draw (4,-3) node[anchor=base]{$\wildcard$};
\draw (5,-3) node[anchor=base]{$\wildcard$};
\draw (6,-3) node[anchor=base]{$\wildcard$};
\draw (7,-3) node[anchor=base]{$\wildcard$};
\draw (8,-3) node[anchor=base]{c};
\draw (9,-3) node[anchor=base]{c};
\draw (10,-3) node[anchor=base]{$\wildcard$};
\draw (11,-3) node[anchor=base]{$\wildcard$};
\draw (12,-3) node[anchor=base]{$\wildcard$};
\draw (13,-3) node[anchor=base]{$\wildcard$};
\draw (14,-3) node[anchor=base]{$\wildcard$};
\draw (15,-3) node[anchor=base]{d};
\draw (6,-4) node[anchor=base]{b};
\draw (7,-4) node[anchor=base]{$\wildcard$};
\draw (8,-4) node[anchor=base]{c};
\draw (9,-4) node[anchor=base]{c};
\draw (10,-4) node[anchor=base]{$\wildcard$};
\draw (11,-4) node[anchor=base]{$\wildcard$};
\draw (12,-4) node[anchor=base]{$\wildcard$};
\draw (13,-4) node[anchor=base]{$\wildcard$};
\draw (14,-4) node[anchor=base]{$\wildcard$};
\draw (15,-4) node[anchor=base]{d};
\draw (6,-5) node[anchor=base]{b};
\draw (7,-5) node[anchor=base]{$\wildcard$};
\draw (8,-5) node[anchor=base]{$\wildcard$};
\draw (9,-5) node[anchor=base]{c};
\draw (10,-5) node[anchor=base]{c};
\draw (11,-5) node[anchor=base]{$\wildcard$};
\draw (12,-5) node[anchor=base]{$\wildcard$};
\draw (13,-5) node[anchor=base]{$\wildcard$};
\draw (14,-5) node[anchor=base]{$\wildcard$};
\draw (15,-5) node[anchor=base]{d};
\draw (18,-6) node[anchor=base]{b};
\draw (19,-6) node[anchor=base]{$\wildcard$};
\draw (20,-6) node[anchor=base]{$\wildcard$};
\draw (21,-6) node[anchor=base]{c};
\draw (22,-6) node[anchor=base]{c};
\draw (23,-6) node[anchor=base]{$\wildcard$};
\draw (24,-6) node[anchor=base]{$\wildcard$};
\draw (25,-6) node[anchor=base]{$\wildcard$};
\draw (26,-6) node[anchor=base]{d};
\draw (0.5,-6.5) -- (28.5,-6.5);
\end{tikzpicture}
\caption{The five occurrences of the query pattern $p= \texttt{b} \gap{0}{4} \texttt{cc} \gap{3}{5} \texttt{d}$ in the string $t=\texttt{acbccbacccddabdaabcdccbccdaa}$.\label{fig:vlgexample}}
\end{figure}

\noindent As shown by \autoref{ex:vlg}, different occurrences of the query pattern $p$ can start or end at the same position in $t$, and the same substring in $t$ can contain multiple occurrences of $p$. Hence to completely characterize an occurrence of $p$ in $t$, we need to report the positions of the individual subpatterns $p_0,p_1,\ldots,p_j$ for each full occurrence of the pattern. However, in the following we will restrict our attention to reporting the start and end position of each occurrence of $p$ in $t$. For the above example, we would thus report the pairs $(3,11)$, $(3,15)$, $(6,15)$ and $(18,26)$.\\

\subsection{Supporting Variable Length Gaps}\label{vlg-howto}
Recall that a variable length gap $\gap{a_i}{b_i}$ is equivalent to $a_i$ wildcards followed by $b_i - a_i$ optional wildcards. Hence to support variable length gaps, we only have to describe how the search algorithms must be modified to match an optional wildcard in $p$. We simulate an optional wildcard as matching both a normal wildcard and the empty string. When matching a normal wildcard the search can only branch in explicit vertices, but for optional wildcards the search will always branch to at least two locations. This is the reason for the $2^{B-A}$ factor in the query times of \autoref{vlg-optimalspaceindex}--\ref{vlg-optimaltimeindex}.

To report the substrings in $t$ where the query pattern occurs, we assume that each leaf $\ell$ in $T(\suff(t))$ has been labeled by the start position, $\ind(\ell)$, of the suffix in $t$ it corresponds to. The search for $p$ terminates in a set of locations $\mathcal R$, each corresponding to one or more substrings in $t$ where the query pattern $p$ occurs. We can report the start and end position of these substrings by traversing the subtrees rooted in the locations of $\mathcal R$. For a subtree rooted in $\ell^\prime \in \mathcal R$ we identify the leaves $\ell_0,\ell_1,\ldots,\ell_r$ corresponding to suffixes of $t$ having $\ell^\prime$ as a prefix. The start and end positions of these substrings are then given by
\[
\left (\ind(\ell_0), \ind(\ell_0)+|\ell^\prime| \right ), \left (\ind(\ell_1), \ind(\ell_1)+|\ell^\prime| \right ), \ldots, \left ( \ind(\ell_r), \ind(\ell_r)+|\ell^\prime| \right ) \; .
\]

\subsection{Analysis of the Modified Search}
To analyse the query time we bound the maximum number of LCP queries performed during the search for the query pattern
\[
p = p_0\; \gap{a_1}{b_1}\; p_1\; \gap{a_2}{b_2}\; \ldots\; \gap{a_j}{b_j}\; p_j \; .
\]
We define $A_i = \sum_{l=1}^i a_l$ and $B_i = \sum_{l=1}^i b_l$. The number of normal and optional wildcards preceding the subpattern $p_i$ in $p$ is $A_i$ and $B_i - A_i$, respectively. To bound the number of locations in which an LCP query for the subpattern $p_i$ can start, we choose and promote $l=0,1,\ldots,B_i-A_i$ of the preceding optional wildcards to normal wildcards and discard the rest. For a specific choice there are exactly $A_i+l$ wildcards preceding $p_i$, and thus the number of locations in which an LCP query for $p_i$ can start is at most $\beta^{A_i+l}$. The term $\beta$ is an upper bound on the branching factor of the search when consuming a wildcard. For a suffix tree $T(\suff(t))$ the branching factor is $\beta=\sigma$, but indexes based on wildcard trees can have a smaller branching factor.
There are $\binom{B_i - A_i}{l}$ possibilities for choosing the $l$ optional wildcards, so the number of locations in which an LCP query for $p_i$ can start is at most
\begin{equation*}
\sum_{l=0}^{B_i-A_i} \binom{B_i-A_i}{l} \beta^{A_i+l} ~\leq~ 2^{B_i-A_i} \beta^{B_i} .
\end{equation*}
Summing over the $j+1$ subpatterns, we obtain a bound of
\(
O \left( 2^{B-A} \beta^B \right)
\)
on the number of LCP queries performed during a search for the query pattern $p$.
%
%
Since LCP queries are performed in time $O(\log \log n)$ and we have to preprocess the pattern in time $O(m)$, the total query time becomes $O(m + 2^{B-A} \beta^B \log \log n + occ)$. This concludes the proof of \autoref{vlg-optimalspaceindex} and \autoref{vlg-wst-index}.

To show \autoref{vlg-optimaltimeindex}, we apply a black-box reduction very similar to \autoref{cuttop}, leading to a $(k,o)$-bounded optional wildcard index, where $k$ and $o$ are the maximum number of normal and optional wildcards allowed in the pattern, respectively. This index consists of the following two optional wildcard indexes. A query is performed on one of these indexes depending on the length $m+B$ of the query pattern $p$.
\begin{enumerate}
\item The unbounded optional wildcard index given by \autoref{vlg-optimalspaceindex}. This index has query time $O(m+2^{B-A}\sigma^B\log\log n + occ)$ and uses space $O(n)$.
\item The $(k,o)$-bounded optional wildcard index obtained by using the wildcard tree $T_1^{k+o} (\pref_G(\suff(t)))$ without the LCP data structure, where $G=\sigma^{k+o}\log\log n$. For $\beta=1$ the search for the subpattern $p_i$ can start from at most $2^{B_i-A_i}$ locations. Searching for $p_i$ from each of these locations takes time $O(|p_i|+b_i)$, since the LCP data structure is not used and the tree must be traversed one character at a time. Summing over the $j+1$ subpatterns, we obtain the following query time for the index
\[
O \left( \sum_{i=0}^j 2^{B_i - A_i} (|p_i| + b_i) + occ \right) ~=~ O \left( 2^{B-A} (m + B) + occ \right) \; .
\]
The index is a wildcard tree and by the same argument as for \autoref{optimaltimeindex}, it can be stored using space $O(n G^{k+o})$.
\end{enumerate}
In case the query pattern $p$ has length $m+B > G$ we query the first index. It follows that $2^{B-A} \sigma^B \log\log n \leq 2^{B-A} G < 2^{B-A}(m+B)$, so the query time is $O(2^{B-A}(m+B)+occ)$. If $p$ has length $m+B \leq G$ all occurrences of $p$ in $t$ can be found by querying the second index in time $O(2^{B-A}(m+B)+occ)$. The space of the index is 
\[
O(n + nG^{k+o}) = O(n (\sigma^{k+o} \log\log n)^{k+o}) = O(n \sigma^{(k+o)^2} \log^{k+o}\log n) \; . 
\]
This concludes the proof of \autoref{vlg-optimaltimeindex}.

\section{Conclusion}
We have presented several new indexes supporting patterns containing wildcards and variable length gaps. All previous wildcard indexes have query times which are either exponential in the number of wildcards or gaps in the pattern, or linear in the length of the indexed text. We showed that it is possible to obtain an index with linear query time while avoiding space usage exponential in the length of the indexed string. Moreover, we gave an index with linear space usage and a fast query time. For wildcard indexes having a query time sublinear in the length of the indexed string, an interesting open problem is whether there is an index where neither the size nor the query time is exponential in the number of wildcards or gaps in the pattern.

\bibliographystyle{abbrv}
\bibliography{references}

\end{document}